\documentclass[
]{ceurart}

\usepackage{listings,caption,subcaption}
\usepackage{mathtools,amsthm,amssymb}
\theoremstyle{definition}
\newtheorem{definition}{Definition}
\theoremstyle{plain}
\newtheorem{theorem}{Theorem}
\newtheorem{lemma}{Lemma}

\theoremstyle{remark}
\newtheorem{example}{Example}
\usepackage{comment}
\usepackage{enumitem}

\DeclareMathOperator*{\argmin}{argmin} 

\newcommand{\ConceptTerm}[1]{\ensuremath{\text{\texttt{#1}}}}
\newcommand{\ConPAtt}{ConPAtt}

\begin{document}

\copyrightyear{2025}
\copyrightclause{Copyright for this paper by its authors.
  Use permitted under Creative Commons License Attribution 4.0
  International (CC BY 4.0).}

\conference{OVERLAY 2025, 7th International Workshop on Artificial Intelligence and Formal Verification, Logic, Automata, and Synthesis, October 26th, 2025, Bologna, Italia}

\title{%
Attack logics, not outputs: Towards efficient robustification of deep neural networks by falsifying concept-based properties
}

\author[1]{Raik Dankworth}[%
orcid=0009-0001-5617-2069,
email=r.dankworth@uni-luebeck.de,
url=https://isp.uni-luebeck.de/staff/r-dankworth,
]
\author[1]{Gesina Schwalbe}[%
orcid=0000-0003-2690-2478,
email=gesina.schwalbe@uni-luebeck.de,
url=https://isp.uni-luebeck.de/staff/g-schwalbe,
]
\address[1]{University of Lübeck, Germany}


\begin{abstract}
Deep neural networks (NNs) for computer vision are vulnerable to adversarial attacks, i.e., miniscule  malicious changes to inputs may induce unintuitive outputs.
One key approach to verify and mitigate such robustness issues is to falsify expected output behavior. This allows, e.g., to locally proof security, or to (re)train NNs on obtained adversarial input examples.
Due to the black-box nature of NNs, current attacks only falsify a class of the \emph{final output},
such as flipping from $\ConceptTerm{stop\_sign}$ to $\neg\ConceptTerm{stop\_sign}$.
In this short position paper we generalize this to search for generally \emph{illogical} behavior,
as considered in NN verification: falsify constraints (\emph{concept-based properties}) involving further human-interpretable concepts,
like $\ConceptTerm{red}\wedge\ConceptTerm{octogonal}\rightarrow\ConceptTerm{stop\_sign}$.
For this, an easy implementation of concept-based properties on already trained NNs is proposed using techniques from explainable artificial intelligence.
Further, we sketch the theoretical proof that attacks on concept-based properties are expected to have a reduced search space compared to simple class falsification, whilst arguably be more aligned with intuitive robustness targets.
As an outlook to this work in progress we hypothesize that this approach has potential to efficiently and simultaneously improve logical compliance and robustness.
\end{abstract}


\maketitle

\section{Introduction}

Neural Networks (NNs) excel in processing subsymbolic inputs like images, and are increasingly being considered for use in safety-critical domains~\cite{rech_artificial_2024}.
This makes it crucial to ensure their robust and intuitive generalization, at least around known training cases.
One tool to evaluate vulnerability to malicious attacks are 
Adversarial Attacks (AAs):
These craft inputs that induce incorrect or unexpected predictions, using minimal modifications to a correctly handled input $\mathbf{x}$ with $\mathbf{y}=f(\mathbf{x})$
\cite{%
suryanto_dta_2022,li_physical-world_2023,hu_adversarial_2023,zheng_physical_2024,
zhu_learning_2024,wang_boosting_2024, ming_boosting_2024
}.
However, existing attacks solely focus on altering the model's final output, i.e., falsify $\forall \mathbf{x}' \in \text{Nbhd}(\mathbf{x})\colon  f(\mathbf{x}') = y$ for some neighborhood $\text{Nbhd}(\mathbf{x})$ around $\mathbf{x}$ like an $\epsilon$-ball. 
This disregards whether the prediction still conforms to high-level, interpretable properties.
Common examples of known properties are sufficient conditions, e.g., $\ConceptTerm{red}(\mathbf{x})\wedge \ConceptTerm{octogonal}(\mathbf{x})\implies \ConceptTerm{stop\_sign}(\mathbf{x})$ in traffic sign recognition from images $x$;
and necessary conditions, like $\neg\ConceptTerm{octogonal}(\mathbf{x})\implies\neg \ConceptTerm{stop\_sign}(\mathbf{x})$.
More general, rules involving unary predicates not available from the NN outputs are here called \emph{concept-based properties}. Rich semantic rules are known to be well suited for runtime plausibility monitoring \cite{%
schwalbe2022enabling,%
giunchiglia2022roadr%
} and respective fixing of NN outputs 
\cite{%
giunchiglia2022roadr,%
ledaguenel2024improving,
badreddine2022logic%
}. 
In particular, they don't constrain the local \emph{output} to be correct, but the underlying \emph{general logical reasoning} locally around the sample.

One reason why falsification of such informative constraints are not considered for attack generation is that they require outputs for all involved predicates---not only the available final output, like $\ConceptTerm{stop\_sign}$.
These however, might need a considerable amount of training data or hyperparameter tuning if added right away during the training; or, even worse, not all properties and thus not all required concepts might be known at training time due to specification gaps or later domain transfer. 

The trick we now use here is that NNs automatically learn to encode task-related concepts in their intermediate outputs. For example, when trained for \ConceptTerm{stop\_sign} recognition, the NN may implicitly learn to identify \ConceptTerm{octagons}, \ConceptTerm{red} color, and the \ConceptTerm{stop\_label}.
Post-hoc supervised concept-based explainability methods~\cite{bau_network_2017,fong_net2vec_2018,crabbe_concept_2022,oikarinen_clip-dissect_2022} can recover this information in a very sample-efficient manner with minimal additions to the NN structure. 

Altogether, we propose and theoretically analyze a general AA goal ---the \emph{Concept-based Property Attack (\ConPAtt)}---that explicitly targets falsification of symbolic \textit{concept-based properties} over non-symbolic inputs.
As we will show, our formulation offers a more general way to define both targeted and untargeted attacks.
Furthermore, as opposed to classical attacks that purely change the output, our attack on $\neg\ConceptTerm{octogonal}\implies\neg\ConceptTerm{stop\_sign}$ can produce an image still classified as \ConceptTerm{stop\_sign}, in which the \ConceptTerm{octogonal} concept is no longer recognized.
This newly allows to uncover failure cases with semantically inconsistent yet possibly high-confidence predictions that are invisible to standard attacks.
As we show, standard white-box attack techniques can still easily and efficiently be applied, producing meaningful attacks and a more constrained adversarial space as compared to traditional AAs.
%
%
%

\paragraph{Contributions.}
Our main contributions are:
\begin{itemize}
    \item We introduce \ConPAtt, a general XAI-supported adversarial attack goal that targets concept-based properties rather than just NN outputs.
    \item We proof that \ConPAtt{} generalizes both classical targeted and untargeted AA formulations, but same-sized or smaller adversarial space.
    \item We hypothesize several advantages of \ConPAtt{}s for certifying robustness and for adversarial retraining, posing the chance to efficiently improve both semantic consistency and robustness.
\end{itemize}


\section{Related Work}

\paragraph{Adversarial Attacks} 
AAs generally search within the vicinity of an input sample $x$ for minimally perturbed variants $\tilde{x}=x+\epsilon$ that have a malicious effect on the NN's output \cite{khamaiseh2022adversarial}.
The perturbations can be arbitrary (digital AAs, considered) \cite{szegedy_intriguing_2014,wu_skip_2019,su_one_2019,wang_enhancing_2021,zhu_learning_2024,wang_boosting_2024}, or
further constrained to realistic changes (physical AAs)
\cite{li_physical-world_2023,hu_adversarial_2023,%
eykholt_robust_2018,liu_perceptual-sensitive_2019,suryanto_dta_2022,zheng_physical_2024}. However, the minimality makes the changes often invisible or difficult to see for humans.
At the methodological level, black-box approaches only require access to NN inputs and outputs \cite{liu_perceptual-sensitive_2019,suryanto_dta_2022,hu_adversarial_2023,huang_safari_2023,wang_rfla_2023,zheng_physical_2024}. White-box attacks as considered here instead exploit NN model internals, such as the gradient, for a more efficient search \cite{szegedy_intriguing_2014,wu_skip_2019,wang_enhancing_2021,li_physical-world_2023,zhu_learning_2024,wang_boosting_2024}.
Generally, AAs can be seen as a subfield of NN verification that falsifies a continuity property \cite{liu2021algorithms,sun2018concolic}.
Thus, usual search, reachability analysis, and---most prominently---optimization techniques are applicable to find or disprove adversarial examples \cite{khamaiseh2022adversarial}.
Regarding types of specifications beyond continuity properties, approaches such as Scenic \cite{fremont_scenic_2019} and VerifAI \cite{dreossi_verifai_2019} demonstrate how formal specifications can be used to generate and analyze simulation-based scenarios with symbolic inputs. In contrast, our approach targets AAs on non-symbolic image inputs, which prevents the direct use of such tools but similarly requires formal specifications.

\paragraph{Concept-based Explainability}
Concept-based explainability generally aims to associate human-interpretable concepts with representations in NN latent space \cite{lee2025conceptbased,poeta2023conceptbased,schwalbe2022concept}.
This includes understanding which concepts are relevant to the decision and to what extent \cite{wan_nbdt_2020,oikarinen_clip-dissect_2022}, and how these can be accurately recognized in NNs \cite{fong_net2vec_2018,maglogiannis_verification_2021}.
If concept definitions in form of labeled samples are available at training time, ante-hoc approaches~\cite{koh_concept_2020,wan_nbdt_2020,yuksekgonul_post-hoc_2022,oikarinen_label-free_2022,yang_language_2023} can train individual neurons to activate for the concept.
We here instead consider post-hoc approaches: These train a simple model to predict the concept of interest from an NN layer's activation \cite{kim2018interpretability}. Other than single-neuron-associations \cite{bau2017network,olah2017feature}, or complex models \cite{crabbe2022concept,zhang2021invertible}, linear models considered here \cite{fong2018net2vec,kim2018interpretability,graziani2018regression,mikriukov2025local} pose a good tradeoff between capturing the entanglement of representations \cite{fong2018net2vec,dreyer2024pure}, interpretability \cite{kim2018interpretability}, and favorably simple representation of the concept as halfspace in the NN's latent space.

\paragraph{XAI and Verification}
Prior work has shown that concept-based explanation methods are vulnerable to adversarial attacks.
Perturbations can mislead attribution \cite{slack_fooling_2020} and concept-based tools \cite{mikriukov_unveiling_2024,brown_making_2023}, and adversarial examples significantly alter the internal concept composition of NNs \cite{mikriukov_unveiling_2024}, confirming the general fragility of interpretability methods \cite{ghorbani_interpretation_2019}.
However, these studies target concepts in isolation, without considering their joint relation to model predictions.

Beyond highlighting vulnerabilities, concept outputs have also been used for verification.
Mangal et al. \cite{mangal_concept-based_2024} employed vision–language models to check concept-based properties.
While expressive, this approach relies on semantic similarity in multimodal embeddings (e.g., CLIP \cite{radford_learning_2021}), which can introduce linguistic ambiguity as well as imprecision for similar terms with small visual differences, e.g., \ConceptTerm{circle} versus \ConceptTerm{octagon}.
Moreover, it is restricted to the latent space of a specific layer, although simple visual concepts may predominantly appear earlier and diminish in later layers.
Cheng et al.~\cite{cheng_towards_2020} proposed specifications close to the output layer, but without decomposing them into underlying concepts and by employing an additional NN.
Semantic losses \cite{xu2018semantic,badreddine2022logic} like logic tensor networks \cite{badreddine2022logic} suggest to directly train concept-based rules into the network. These techniques, however, are only used for updating the NN, not for verification as done in \cite{schwalbe2022enabling}, and not for AAs. Furthermore, they rely on concepts being direct outputs of the NN. Even further decoupling the verification from the NN's learned representations and thus exacerbating training efforts, Xie et al.~\cite{xie2022neurosymbolic} even trained completely separate NNs for predicting the concepts. 
Our work also directly addresses the relationship between concepts and model outputs like, a perspective that has received little attention so far \cite{schwalbe2022enabling}. However, similar to the verification testing techniques from \cite{cheng_towards_2020,schwalbe2022enabling}, we suggest to keep training and verification efforts low by using faithful explainability techniques to access concept predictions, and we newly apply the setup to AAs.

\section{Background}

\paragraph{Adversarial Attacks}
Let $\mathbf{x} \in \mathcal{X}$ be a real image, $\mathbf{y} \in \mathcal{Y}$ be its true label, and $f\colon \mathcal{X} \to \mathcal{Y}$ be a NN.
An AA seeks an adversarial example $\mathbf{x}^{\text{adv}} \coloneqq \mathbf{x} + \epsilon \in \mathcal{X}$ so that its output is (sufficiently) different from the original, 
and the perturbation~$\epsilon$ is minimal to an objective function $o$ (usually the L1, L2, or L-infinity norm on the input for digital attacks).
Sufficient difference can be formulated in terms of a $\mathbf{y}$-specific partition of the output set $\mathcal{Y}$ into a benign output set
$\mathcal{Y}^+
\subset \mathcal{Y}$ with $f(\mathbf{x})\in\mathcal{Y}^+$, and a malicious one $\mathcal{Y}^- \coloneqq \mathcal{Y} \setminus \mathcal{Y}^+$. 
The search for the minimum perturbation~$\epsilon$ then is the optimization problem
\begin{align}
    &{\textstyle\argmin_\epsilon} \> o(\epsilon)  \qquad
    \text{s.t. }  f(\mathbf{x} + \epsilon) \in \mathcal{Y}^-
    \;.
\end{align}
%
%
%
Adversarial attack strategies for classification are categorized as \emph{targeted} or \emph{untargeted} according to their choice of $\mathcal{Y}^{-}$:
Let $p_l \colon \mathcal{Y} \to [0,1]$ denote the confidence assigned to class~$l$, and $\theta_l \in [0,1]$ the threshold required to accept class~$l$.
In untargeted attacks, the goal is to reduce the confidence of the true class below threshold, i.e., $\mathcal{Y}^- = \{ \mathbf{y} \in \mathcal{Y} \mid p_l(\mathbf{y}) < \theta_l \}$.
In contrast, targeted attacks aim to raise the confidence of an incorrect class~$l'$ above a threshold, i.e., $\mathcal{Y}^- = \{ \mathbf{y} \in \mathcal{Y} \mid p_{l'}(\mathbf{y}) \geq \theta_{l'} \}$.

\paragraph{Post-hoc Concept Extraction} 
Let $C$ be a set of concepts (e.g., $C=\{\ConceptTerm{red}, \ConceptTerm{orthogonal}\}$), and assume a possibly small classification dataset $\mathcal{D}_c = \left((\mathbf{x}_k, \mathbf{y}_{c,k})\right)_k$ is available per concept $c\in C$. Further denote by $f_{i\to j}\colon \mathcal{X}_i\to\mathcal{X}_j$ the NN part that maps from the $i$th to the $j$th layer.
Through linear post-hoc concept extraction, additional concept outputs are added to the NN by attaching for each $c$ a linear classification model $f_{i\to c}\colon \mathcal{X}_i\to \mathcal{C}_c=[0,1]$ to the $i$th hidden layer as illustrated in \autoref{fig:post-hoc-network}. Keeping the NN's weights fixed, the weights of $f_{i\to c}$ are trained on pairs $\left((f_{\to i}(\mathbf{x}_k), \mathbf{y}_{c,k})\right)_k$,
such that $c$'s concept function $f_c = f_{i\to c}\circ f_{\to i}\colon \mathcal{X}\to \mathcal{C}_c$ correctly predicts presence of the concept in an input image.
Note that $f_{i\to c}$ being linear conveniently makes any subspace $\{v\in \mathcal{X}_i\mid f_{i\to c}(v) > \theta\}$ an affine linear half-space.
In the following, we denote by $f_C = (f_c)_{c\in C}\colon \mathcal{X} \to \mathcal{C} = (\mathcal{C}_c)_{c\in C} $ the complete prediction of all concepts, and by $\mathcal{Z} = \mathcal{Y} \times \mathcal{C}$ the complete output set after attaching the concept outputs.


\begin{figure}
    \centering
    \vspace*{-2.0\baselineskip}%
    \includegraphics[width=.45\linewidth]{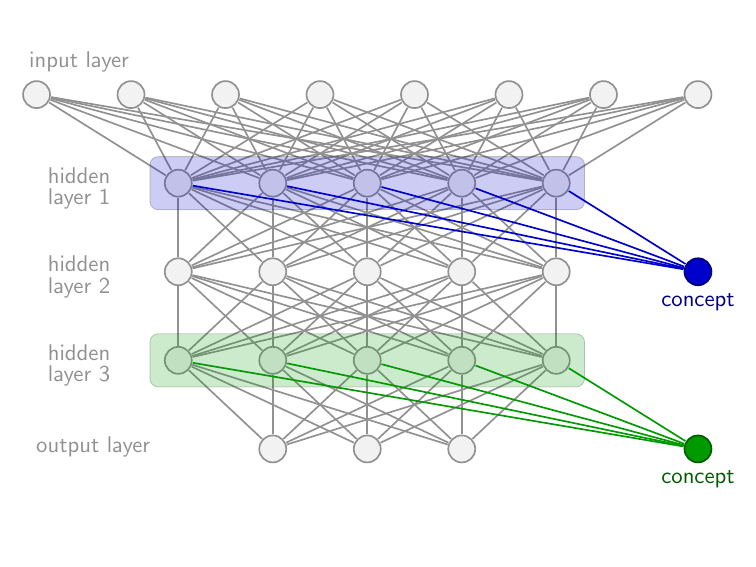}%
    \\\vspace*{-2.0\baselineskip}%
    \caption{\textbf{Post-hoc Concept Extraction:}
    Two neurons for concept output (\textcolor{blue}{blue} and \textcolor{green}{green} neurons) are post-hoc added to the trained NN (\textcolor{gray}{gray}) using newly trained connections to hidden layers \textcolor{blue}{1}/\textcolor{green}{3}.
    }
    \label{fig:post-hoc-network}
\end{figure}

\paragraph{T-Norm Fuzzy Logic}
\label{subsec:fuzzy-logic}
The standard Boolean logical connectives (\emph{and} $\wedge$, \emph{or} $\vee$, \emph{not} $\neg$) can only operate on binary truth values in $\mathbb{B}=\{0,1\}$. T-norm fuzzy logics extend the connectives to many-valued truth values in $\mathbb{B}=[0,1]$ using a so-called t-norm $\wedge_t\colon [0,1]\times [0,1]\to [0,1]$ to replace the $\wedge$. A valid t-norm must be monotonic, commutative and associative, have a neutral element (the $1$), and match $\wedge$ on Boolean values.
Typical choices for $a\wedge_t b$ are Product ($a\cdot b$), Łukasiewicz ($\max(0,a+b-1)$), and Gödel ($\min(a,b)$) t-norms \cite{hajek_metamathematics_1998}, since these form a generating system for all continuous t-norms.
Given a $\wedge_t$, then $\neg a \coloneqq 1-a$, $\vee_t,\implies_t\colon[0,1]^2\to[0,1]$ can be derived and maintain desirable properties, giving the resulting t-norm logic.

Desirable properties for use of t-norm logic with NN classification outputs are:
(1) The NN typically produces a confidence prediction in $[0,1]$ instead of a Boolean value, which can be propagated by t-norm fuzzy logic to the confidence of entire logical expressions.
(2) The classicale piece-wise continuous t-norm logic connectives are also piece-wise differentiable like ReLU activations of NNs. So, they can directly be used in backpropagation \cite{badreddine2022logic}.

\section{Approach}
In this chapter we first define our new notion of concept-based AAs. Then we show that standard AAs are a special case, and existing attack techniques can easily be adopted to our new attack.

\subsection{Concept-based Property Attacks}

These classical AA types can also be interpreted as special cases of \emph{property attacks}, where class predictions are treated as logical literals.
Using fuzzy logic (see \autoref{subsec:fuzzy-logic}), we can evaluate logical expressions over outputs using a function $\text{solve}_\rho \colon \mathcal{Y} \to \mathbb{B}$ that returns the truth value of a property~$\rho$.
A property attack falsifies a given property, i.e., $\mathcal{Y}^- = \mathcal{Y}^-_\rho = \{ \mathbf{y} \in \mathcal{Y} \mid \text{solve}_{\lnot \rho}(\mathbf{y}) \}$.
Untargeted and targeted attacks correspond to properties $\rho = l$ and $\rho = \lnot l$, respectively.

This perspective allows adversarial examples to be crafted with higher-order conditions — e.g., enforcing both ``dog''~($d$) and ``cat''~($c$) simultaneously.
The corresponding attacked property is its logical negation: $\rho = \lnot d \lor \lnot c$.

The point of view of property attacks can also be applied to NNs that are boosted with XAI techniques.
The additional concept outputs can also be used as well as the original task output of the NN to define property attacks--- \emph{Concept-based Property Attacks} (\ConPAtt{}).
For denoting the properties we propose to use the following intuitive and convenient implication format generalizing our introductory examples (all logical expressions can be reformulated like this, see \autoref{lemma:implication-form}).
Note that for simplicity we shorten $f_c(\mathbf{x})$ to $c$, and $(\neg)c$ shorthands possibly negated $c$.

\begin{lemma}
    \label{lemma:implication-form}
    Each  logical expression~$\varphi$ with two disjoint literal sets $C$ and $L$ can be reformulated into a term of conjunctively linked implication terms where antecedents consist only of conjunctively linked, possibly negated literals of $C$, and consequences consist only of disjunctively linked, possibly negated literals of $L$.
\end{lemma}

\begin{proof}
Each logical expression can be reformulated into the conjunctive normal form
$\textstyle
    \varphi \equiv
    \bigwedge_{i > 0}
    (
    \bigvee_{c \in C_i \subseteq C} (\lnot) c 
    \bigvee_{l \in L_i \subseteq L} (\lnot) l
    )
$.
%
Let us introduce two additional variable families $\alpha_i, \beta_i$ that condense the disjunctive subformulas:
\begin{align}
    \alpha_i &\coloneqq \lnot \bigvee\limits_{c \in C_i \subseteq C} (\lnot) c  
    \equiv  \bigwedge\limits_{c \in C_i \subseteq C} (\lnot) c &
    \beta_i &\coloneqq \bigvee\limits_{l \in L_i \subseteq L} (\lnot) l
\end{align}
The subformulas can be replaced by these variables and the whole logical expression~$\varphi$ reformulates to
$
    \varphi 
    \equiv \bigwedge\limits_{i > 0} \left(\lnot \alpha_{i} \lor \beta_{i}\right)
    \equiv \bigwedge\limits_{i > 0} \left(\alpha_{i} \implies \beta_{i}\right)
$.
\end{proof}

\begin{definition}[Concept-based property]
\label{def:concept-based-property}
A concept-based property~$\varphi$ is a logical expression with two disjoint literal sets $C$---the concept literals---and $L$---the task literals---in the form of conjunctively linked implication terms whose antecedents consist only of conjunctively linked, possibly negated concept literals and whose consequences consist only of disjunctively linked, possibly negated task literals.
\begin{align}
    \varphi \coloneqq \bigwedge\limits_{i > 0} \left(\alpha_{i} \implies \beta_{i}\right)\;,
    \quad\text{ with }\quad
    \alpha_i \coloneqq \bigwedge\limits_{\mathclap{c \in C_i \subseteq C}} (\lnot) c
    \quad\text{ and }\quad
    \beta_i \coloneqq  \bigvee\limits_{\mathclap{l \in L_i \subseteq L}} (\lnot) l
\end{align}
\end{definition}

\begin{definition}[Concept-based Property Attack]
  Let $\text{solve}_\varphi \colon \mathcal{Z} \to \mathbb{B}$ be the function to calculate the truth value of a concept-based property~$\varphi$ which evaluates to true at an input $\mathbf{x}$, and $o$ a minimality measure for perturbations $\epsilon$.
  A Concept-based Property Attack of $\varphi$ is the search for a $o$-minimal perturbation $\epsilon$ to an input $\mathbf{x}$ into an adversarial example~$\mathbf{x}^{\text{adv}}=\mathbf{x}+\epsilon$ which falsifies $\varphi$, i.e.,
  lies in the malicious output set
  \begin{align}
    \SwapAboveDisplaySkip
    \mathcal{Z}^- = \mathcal{Z}^-_\varphi =  \lbrace \mathbf{z} \in \mathcal{Z} \mid  \text{solve}_{\lnot \varphi}(\mathbf{z}) \rbrace
  \end{align}
\end{definition}

Intuitively, a \ConPAtt{} adversarial example $\mathbf{x}^{\text{adv}}$ to $\phi=(\bigwedge_{c}c\implies\bigvee_{l}l)$ like $\ConceptTerm{red}\wedge \ConceptTerm{octogonal}\implies \ConceptTerm{stop\_sign}$, causes the NN to predict all $c$ as true, and all $l$ as false.
This can happen if (1) some $c$ is predicted true even though it should be false (e.g., \ConceptTerm{red} predicted true even though the change $\epsilon$ turned the sign \ConceptTerm{gray}), and/or (2) some $l$ is predicted negative even though it should be positive (e.g., \ConceptTerm{stop\_sign} flipped to false).

\subsection{\ConPAtt{}s as Generalized Adversarial Attacks}
Note that falsifying one implication term is enough to falsify a concept-based property and thus, it is sufficient to consider one implication $\phi=\alpha\implies \beta$ for an attack.
The set of adversarial example task outputs can be derived from this definition, i.e. $\mathcal{Y}^-_{\varphi} \coloneqq \lbrace \mathbf{y} \in \mathcal{Y}  \mid (\mathbf{y}, \mathbf{c})  \in \mathcal{Z}^-_{\varphi} \rbrace$. Furthermore:
\begin{theorem}\label{thm:special-cases}
  Standard targeted and untargeted AAs are special cases of \ConPAtt{}.
\end{theorem}
\begin{proof}
    First note the two special cases of \ConPAtt{} where only a single task literal is used:
    \begin{enumerate}[nosep, noitemsep]
        \item \textbf{Generalized untargeted AAs}: $\alpha \implies l$.
        \item \textbf{Generalized targeted AAs}: $\alpha \implies \lnot l$.
    \end{enumerate}
    Un-/targeted respective are generalized un-/targeted AAs with $\alpha \equiv \text{true}$, i.e., no concept restriction.
\end{proof}

A neat property of \ConPAtt{}s is that the search space is generally reduced compared to vanilla AAs:
\begin{theorem}\label{thm:reduced-search-space}
  The task output spaces of adversarial examples for generalized untargeted/targeted AAs are smaller than or equal to those for standard untargeted/targeted AAs.
  \begin{align*}
  \SwapAboveDisplaySkip
    \mathcal{Y}^-_{\alpha \implies l} \subseteq \mathcal{Y}^-_{l} \qquad \mathcal{Y}^-_{\alpha \implies \lnot l} \subseteq \mathcal{Y}^-_{\lnot l}
  \end{align*}
\end{theorem}
\begin{proof}
Let us first look at generalized untargeted AA properties like $\alpha \implies l$.
Each adversarial example must lack class prediction~$l$ but requires concept predictions~$\alpha$, i.e., they satisfy the property $\alpha \land \lnot l$.
In contrast to that, standard untargeted AAs only require the misclassification of $l$, i.e. each adversarial example satisfies $\lnot l$ and they accept adversarial examples that do not additionally fulfill $\alpha$.
It follows that the valid output space of adversarial examples for generalized untargeted AAs~$\mathcal{Z}^-_{\alpha \implies l}$ is smaller than or equal to that for standard untargeted AAs~$\mathcal{Z}^-_{l}$ as well as for their valid task output spaces $\mathcal{Y}^-_{\alpha \implies l} \subseteq \mathcal{Y}^-_{l}$.

In this explanation, it does not matter whether both adversarial examples expect a misclassification~$\lnot l$ or a specific task output~$l$.
That is why this relation also applies between generalized targeted AAs and standard targeted AAs, i.e. $\mathcal{Y}^-_{\alpha \implies \lnot l} \subseteq \mathcal{Y}^-_{\lnot l}$.
\end{proof}

\paragraph{\ConPAtt{} Procedure}
\ConPAtt{} can be easily performed with any existing AA approach.
The trick is to use the result of the (partially) differentiable fuzzy operation $\varphi\circ (f, f_C)\colon\mathcal{X}\to \mathbb{B}$ instead of the output of the NN.
This makes \ConPAtt{} a targeted AA with the expected result False or 0 for the adversarial examples.

\section{Discussion and Outlook: \ConPAtt{} for Adversarial Training}
In the following we discuss further what practical benefits we expect from this more general formulation of attack goals, how this could be evaluated, and which challenges are still open.

\subsection{Hypothesized Benefits of \ConPAtt{}s}

We hypothesize that
\begin{itemize}[nosep,noitemsep]
    \item \emph{generalized (un-)targeted AAs with at least one concept \textbf{reduce the search space} for adversarial examples not only theoretically but also empirically},
    \item \emph{the adversarial examples obtained via \ConPAtt{} are \textbf{particularly efficient for retraining} because they are pinpoint adversarial examples with a high information content}. 
\end{itemize}

\paragraph{\ConPAtt{}s versus Standard AAs:}
To understand above claims, one should first have a closer look at the vulnerabilities that can be exploited for a successfull \ConPAtt{} attack against a concept-based property $\phi=(\alpha\implies l)$. Standard AAs capture any cases, where the final output $l$ is changed, regardless of whether this resulted in illogical behavior breaking $\phi$ or not. Thus, standard AAs may primarily focus on turning off causally related early-layer concepts, i.e., falsifying $\alpha$ to falsify $l$. For example, falsify \ConceptTerm{red} to cause a negative output of \ConceptTerm{stop\_sign}.
This is not sufficient for a \ConPAtt{} to $\phi$, for which not only $l$ must become false, but simultaneously $\alpha$ must remain true (cf.\ \autoref{thm:reduced-search-space}). It is therefore not guaranteed that one obtains the same results for \ConPAtt{}s against any of the following concept-based properties:
\begin{itemize}[nosep]
  \item $\phi_l = (\text{true}\implies l)$, which is the standard AA against the output $l$,
  \item $\phi_\alpha = (\neg\alpha\implies \text{false})$, which is the standard AA against the concept outputs, i.e., the attack flips any concept $c$ in the conjunction $\alpha=\bigwedge_c c$ to false, and
  \item $\phi = (\alpha\implies l)$, which is a generalized concept-based property attack.
\end{itemize}
Whether the obtained adversarial examples are similar depends on whether it is easier to attack concepts---then falsifying $\phi_l$ and $\phi_\alpha$ should yield similar results---or logics, in which case falsifying $\phi_l$ and $\phi$ are expected to yield similar results. Since concepts themselves represent noisy variables with non-perfect accuracy, chances are high that attacking concepts generally is easier than attacking logics.
Our \ConPAtt{} framework provides the option to test and train on these different rules individually, and hence distinguish more finegrained between simply attacking the concepts or outputs, and truly attacking internal logics.


\paragraph{Benefits of Targeting Logics:}
One reason for both of the claims is on semantic level: Human-defined properties typically encode important knowledge about the task at hand, thus should strengthen both the adherence to the properties and indirectly the actual main task of the network.
Given that well-generalizing NNs typically adopt this knowledge to large extend, the cases of logic breaches should be few but meaningful.
This would make \textbf{attacking logics especially beneficial for retraining purposes} similar to adversarial training \cite{khamaiseh2022adversarial,tramer_ensemble_2018}.

\paragraph{Benefits for Computational Efficiency:} Also, here directly benefit from low integration overhead:
(1) Preparation only requires cheap post-hoc concept extraction;
(2) Only very few additional operations (the $f_{i\to c}$) are added that need backpropagation/-tracing if gradient-based attack methods are used; and
(3) The beneficial formulation of concepts as half-spaces in latent spaces allows efficient reachability analysis with substantial reduction in the search space as illustrated in \autoref{fig:concept-propagation} and sketched in \autoref{app:reachability-search-consideration}.
Next steps should empirically test the attack success and the effect of retraining with adversarial examples of this approach.

\begin{figure}[t]
    \centering
        \centering%
        \vspace*{-1.5eM}%
        \includegraphics[width=.7\linewidth]{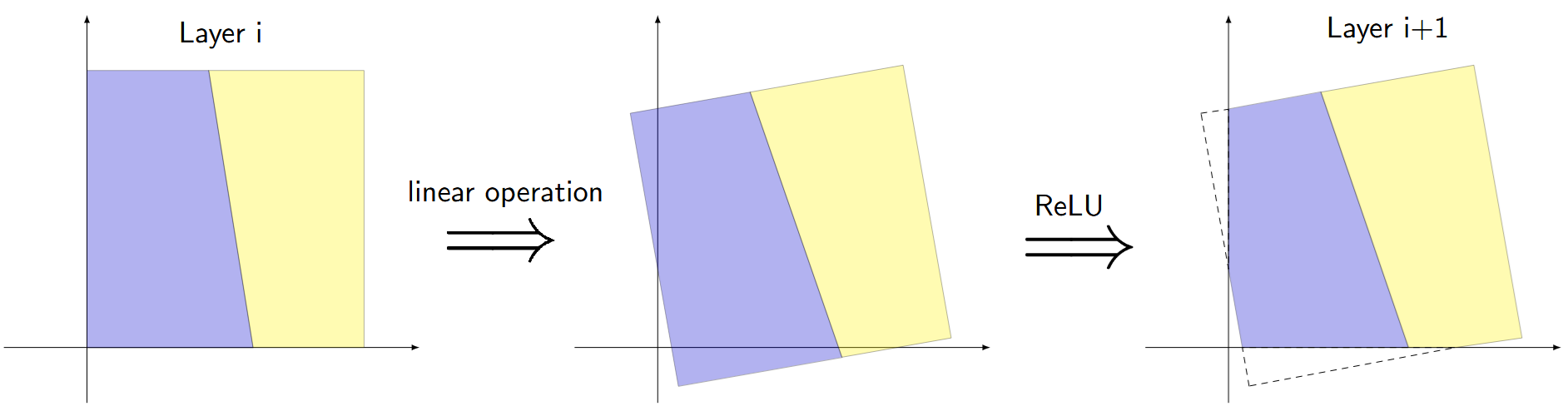}%
        \vspace*{-.5eM}%
        \caption{
        \textbf{Concept Propagation:}
        Illustration how the concept and non-concept half-spaces propagate from layer to layer using linear operations like convolutions (left to mid) combined with ReLU activation (mid to right). Concretely, ReLUs add additional bends of wide angle to the decision boundary.}
        \label{fig:concept-propagation}
\end{figure}

\subsection{Future Work: Evaluation and Challenges}

\paragraph{Planned Experimental Setting:}
We suggest to evaluate several aspects to ensure a comprehensive assessment.
As metrics, we consider both task performance and rule adherence, measured through accuracy and Intersection-over-Union (IoU) for task prediction as well as rule satisfaction.
In addition, we track the success of adversarial attacks before retraining, as well as the effectiveness of defences and the accuracy of concepts after retraining.
For evaluation, we draw on three established \textbf{datasets}: MNIST \cite{deng_mnist_2012}, GTSRB \cite{Stallkamp2012}, and ImageNet \cite{russakovsky_imagenet_2015}.
The \textbf{models} include self-trained simple architectures for MNIST and GTSRB, as well as a range of widely used ImageNet classifiers:
Inception-v3  \cite{szegedy_rethinking_2016},
Inception-v4 \cite{szegedy_inception-v4_2017},
Inception-Resnet-v2 \cite{szegedy_inception-v4_2017},
Resnet-v2-101 \cite{he_deep_2016}
and the ensemble-based variants 
Inception v3$_{ens3}$,
Inception v3$_{ens4}$
and IncRes v2$_{ens}$ \cite{tramer_ensemble_2018}.
For baselines, we rely on several state-of-the-art adversarial attack methods,
namely SGM \cite{wu_skip_2019},
VMI-FGSM and VNI-FGSM \cite{wang_enhancing_2021},
L2T \cite{zhu_learning_2024},
and BSR \cite{wang_boosting_2024}.

The attacked concept-based properties reflect both simple and more complex relations.
Examples include that class \ConceptTerm{1} implies the concept \ConceptTerm{line}, that classes \ConceptTerm{1} and \ConceptTerm{2} should never be predicted simultaneously (i.e., $\lnot 1 \lor \lnot 2$), and that the concepts \ConceptTerm{red}, \ConceptTerm{octagon}, and \ConceptTerm{stop\_label} together imply \ConceptTerm{stop\_sign}.

\paragraph{Challenges and further Future Work:}
%
%
As explained above, it is expected that \ConPAtt{}s not necessarily yield the same results as standard AAs that attack outputs or concepts. In addition to above experiments, one could contrastively compare results for the different attacks for insights how large the gap truly is.
%
%
%
%
%
%
%
%
%
However, a considerable challenge for the experimental evaluation is that retraining procedures may need to be adapted:
(Adversarially) retraining with respect to the \emph{task output} might accidentally destroy the post-hoc attached \emph{concept outputs}. Countermeasures might be to freeze earlier NN parts up to the concept prediction, or alternatingly or simultaneously retrain the NN and the concept predictors. Experiments must show how to balance need for concept labels with concept accuracy during adversarial finetuning.

\section{Conclusion}

In this position paper, we introduce a novel generalized adversarial attack goal: Instead of targeting a change in (respectively falsification of) the output class, our attacks aim to falsify the compliance of the NN with prior symbolic knowledge on sufficient indicators for an output class.
Standard AAs are shown to be a specific case of our generalized formulation for concept-based properties.
Also, these allow to substantially reduce the expected search space of the AA search with increasing number of concepts.
Also, we argue that these concept-based properties provide a more natural and human-aligned target for AAs. This suggests that they might be particularly suited for NN robustification via adversarial model (re)training or runtime monitoring.

\begin{acknowledgments}
    This work was supported through the junior research group project “chAI” funded by the German Federal Ministry of Research, Technology and Space (BMFTR), grant no. 01IS24058. The authors are solely responsible for the content of this publication.
\end{acknowledgments}

\section*{Declaration on Generative AI}
 During the preparation of this work, the author used ChatGPT based on GPT-4o in order to: Improve writing style. After using these tool(s)/service(s), the author reviewed and edited the content as needed and takes full responsibility for the publication’s content.


\bibliography{main}

\newpage
\appendix

\section{Considerations for Reachability-based Search}
\label{app:reachability-search-consideration}
Existing reachability-based techniques conduct forward and/or backward passes through the NN to trace / estimate regions of interest through the NN processing.
We here show how the considered concept-based properties give rise to a particularly efficient formulation of this approach:
Being half-spaces in intermediate layers, the (negated) concepts have the potential to easily and substantially reduce the adversarial space that one needs to keep track off half-way through the network and can also be easily described in later layers as sketched in \autoref{fig:concept-propagation}.
In the following, this is illustrated for a back-propagation approach for a simple generalized untargeted attack $(c_1\wedge \dots \wedge c_n)\implies l$.
Recall that a valid counterexample falsifying the property $(c_1\wedge \dots \wedge c_n)\implies l$ must fulfil $c_1\wedge \dots \wedge c_n \wedge \neg l$.

Denote by $f_{\mathcal{L}\to\mathcal{L}'}\colon \mathcal{L}\to\mathcal{L}'$ the NN part mapping from layer $\mathcal{L}$ to $\mathcal{L}'$, and $p^{(\mathcal{L})}_{\mathcal{L'}\to c}=p_c\circ f_{\mathcal{L}\to\mathcal{C}}\circ f_{\mathcal{L}'\to\mathcal{L}}\colon \mathcal{L'}\to [0,1]$ the function evaluating the presence of concept $c$ in layer $\mathcal{L}$ for a latent vector $v$ from an earlier layer $\mathcal{L}'$.
Denote by $\mathcal{L}_c$ the layer which was chosen for the embedding of concept $c$, and let $\mathcal{L}_1$ be the earliest layer for which $\mathcal{L}_1=\mathcal{L}_c$ for some $c$. Note that $\mathcal{L}_l=\mathcal{L}_{L-1}$ is the final representation layer before the output confidence prediction, if this is $L$ layers later than $\mathcal{L}_1$.
Let $H_{c}=\{v\in \mathcal{L}_c\mid p_{\mathcal{L}_c\to c}(v)<\theta_c\} $ be the halfspace of the concept $c$ in the concept's $\mathcal{L}_c$.

Now we can reformulate the falsification as a search for a region in latent space:
\begin{lemma}
    A representation $v=f_{\to \mathcal{L}}(x)\in \mathcal{L}$ in layer $\mathcal{L}$ of a valid counterexample $x\in\mathcal{X}$ to the concept-based property 
    $(c_1\wedge \dots \wedge c_n)\implies l$
    must fulfil
    $v\in \bigwedge\limits_{c \in c_i, l} f^{-1}_{\mathcal{L}\to\mathcal{L}_c}(H_c)$.
\end{lemma}
\noindent
While it is costly to determine $f^{-1}_{\mathcal{L}\to\mathcal{L}_c}(H_c)$ independently,
the concept-based property gives rise to a recursive definition:
\begin{theorem}
Recursively define the propagation of halfspace intersections through the NN
\begin{align}
P_{L-1} &= \bigcap_{\mathcal{L}_L=\mathcal{L}_c} H_c
\;,&
P_{i} &= f^{-1}_{\mathcal{L}_{i-1}\to\mathcal{L}_i}(P_{i+1})\cap \bigcap_{\mathcal{L}_i=\mathcal{L}_c} H_c
\end{align}
Then for any counterexample $x$ to above concept-based property it must hold that $f_{\to\mathcal{L}}(x)\in P_{1}$.
$P_{1}$ can be efficiently calculated using \emph{a single backward propagation} through layers $L-1$ to $1$.
\end{theorem}
\begin{proof}
    The property inductively follow from the definition, noting that $P_{1} = \bigwedge\limits_{c \in c_i, l} f^{-1}_{\mathcal{L}\to\mathcal{L}_c}(H_c)$ and the constraint of considering ReLU networks.
\end{proof}
In particular, each propagation step only requires to obtain a polytope's preimage for a single NN layer operation, and apply a cheap intersection of the resulting polytope with halfspaces.
This makes the first part of the search very efficient, promising speedup compared to a full end-to-end search for counterexamples directly in the input space.

The forward-propagation case is similar. Here, it can additionally be shown, that the propagated $P_i$ always is a connected polytope, since intersection with half-spaces does not change this property, neither does the forward pass through continuous layer operations.

\end{document}